\documentclass[11pt]{article}

\usepackage[T1]{fontenc}
\usepackage[utf8]{inputenc}
\usepackage[margin=1in]{geometry}
\usepackage{authblk}
\usepackage{amsthm}
\usepackage{amsmath,amssymb}
\usepackage{paralist}
\usepackage{xcolor}
\usepackage{comment}
\usepackage{stmaryrd}
\usepackage{txfonts}
\usepackage[hidelinks]{hyperref}

\DeclareMathOperator{\EN}{\mathsf{EN}}

\DeclareMathOperator{\supp}{\mathsf{supp}}

\mathchardef\mhyphen="2D

\newtheorem{proposition}{Proposition}

\newtheorem{corollary}{Corollary}
\newtheorem{construction}{Construction}
\newtheorem{example}{Example}
\newtheorem{remark}{Remark}

\newtheorem{definition}{Definition}

\newcommand*{\multisetminus}{\setminus\mskip-10mu\setminus}

\title{A Unified Framework for Reaction Systems Based on Interval Structures}

\author[1]{Paolo Bottoni\thanks{Corresponding author: \href{mailto:bottoni@di.uniroma1.it}{bottoni@di.uniroma1.it}}}
\author[1]{Anna Labella\thanks{\href{mailto:labella@di.uniroma1.it}{labella@di.uniroma1.it}}}
\author[2]{Ion Petre\thanks{\href{mailto:ion.petre@utu.fi}{ion.petre@utu.fi}}}

\affil[1]{Department of Computer Science, Sapienza University of Rome, Viale Regina Elena 295, 00161 Rome, Italy}
\affil[2]{Department of Mathematics and Statistics, University of Turku, Vesilinnantie 5, 20014 Turku, Finland}

\date{}

\begin{document}

\maketitle

\begin{abstract}
Reaction systems have evolved into a rich family of computational models differing in their treatment of multiplicities, resource management, concurrency, and state evolution. We introduce a unified semantic framework based on interval structures and interval-based transformation systems. The framework decomposes operational semantics into independent resource, production, update, and execution strategies, providing a common basis for describing, comparing, and constructing reaction-system variants. We show that classical reaction systems, restricted reaction systems, multiset reaction systems, reaction systems with concentration, and resource-preserving multiset reaction systems are all recovered as instantiations of the framework. Quantitative reaction systems are accommodated through an additional preprocessing stage. We further demonstrate that the framework naturally extends beyond reaction systems to other computational models, including Petri nets. The proposed framework provides a common semantic foundation for existing models and a flexible basis for developing and analysing new computational formalisms.
\end{abstract}

\section{Introduction}\label{sec:intro}
Reaction systems were introduced in \cite{ehrenfeucht_reaction_2007} as a qualitative computational model inspired by the functioning of biochemical reaction networks. In contrast to traditional models based on resource consumption, reaction systems emphasize the mechanisms of facilitation and inhibition, capturing the intuition that the occurrence of a reaction depends on the presence of certain entities and the absence of others. Since their introduction, reaction systems have developed into an active research area within natural computing, giving rise to a substantial body of theoretical results and numerous extensions \cite{ehrenfeucht_reaction_2016}.

Several variants of reaction systems have been proposed in order to address limitations of the original qualitative framework. Classical reaction systems operate on sets and therefore only distinguish between presence and absence of entities. To represent quantitative phenomena, various multiset-based and quantitative extensions have been introduced, allowing explicit multiplicities, resource-sensitive behaviour, accumulation of products, saturation effects, and persistence of resources across computational steps \cite{mitrana_quantitative_2025, bottoni_multiset_2025, brodo_quantitative_2023, brodo_exploiting_2021}. 
Additional developments have explored alternative notions of concurrency \cite{koutny_asynchrony_2021}, enabling conditions \cite{ehrenfeucht_evolving_2017}, and execution strategies \cite{manzoni_simple_2014,salomaa_minimal_2013}. 
Hence, the field now contains a rich collection of related models that share common intuitions but differ substantially in their operational semantics.

Despite these developments, a general framework is missing which could support the description and comparison of these variants within a common formal setting. 
Indeed, existing models are typically introduced independently, using different mathematical representations and semantic assumptions. 
This makes it difficult to identify precisely which aspects of the semantics are responsible for the behavioural differences between models, and obscures the relationships between existing approaches.

This paper introduces a unified semantic framework for reaction-system variants. We use interval structures as a common representation for enabling conditions and production outcomes, allowing exact multiplicities, threshold conditions, saturation constraints, and nondeterministic production choices to be expressed uniformly. Building on this representation, we define interval-based transformation systems whose semantics is determined by independent choices of resource, production, update, and execution strategies. This decomposition makes the semantic assumptions underlying existing models explicit, provides a systematic basis for comparing them, and supports the definition of new variants by combining semantic policies.
%
%
%
In particular, we
show that classical reaction systems, restricted reaction systems, multiset reaction systems, resource-preserving multiset reaction systems, and reaction systems with concentration are recovered as particular instantiations of the framework. Quantitative reaction systems are accommodated through an additional preprocessing stage, while the same semantic framework naturally extends beyond reaction systems to 
%
%
other models of dynamic systems, namely Petri Nets and finite state automata.
%
%
%

The paper is organized as follows. 
Section \ref{sec:background} recalls basic notions on multisets. 
Section \ref{sec:S_intervals} introduces interval structures and studies their fundamental properties, in order to define interval transformations and interval transformation systems in Section \ref{S_transformations}. 
Section \ref{S_policies} develops the semantic framework through resource, production, and update policies. 
Section \ref{S_execution} discusses execution strategies and different notions of concurrency, based on which Section \ref{S_variants} revises several existing variants of reaction systems (as well as other dynamical systems) from the framework perspective.
Section \ref{sec:concls} concludes the paper and  points to future research.

\section{Background}\label{sec:background}
%
%
We denote the set of natural numbers by $\mathbb{N}$ and the set of positive natural numbers by $\mathbb{N}^+$. 
We also add an extra symbol $\top$ with the convention that $x<\top$ for all $x\in\mathbb{N}$.

The \textit{cardinality} of a finite set $X$ is denoted by $|X|$.

For a function $f:X\rightarrow{Y}$, we call $X$ the \textit{domain} of $f$, noted $dom(f)$; $Y$ the \textit{codomain} of $f$, noted $cod(f)$; and $\{y\in{Y}\mid{y}=f(x)\mbox{ for some }x\in{X}\}$, the \textit{range} of $f$, noted $rng(f)$.

\begin{definition}[Multiset]\label{def:multiset}
Let $S$ be a finite set. 
A \emph{multiset} over $S$ is a function $M:S\rightarrow\mathbb{N}$, where $M(x)$ denotes the \emph{multiplicity} of $x$ in $M$.
The set of multisets over $S$ is denoted by $\mathcal{M}(S)$. For $M^\prime,M^{\prime\prime}\in\mathcal{M}(S)$, we say that $M^\prime\sqsubseteq M^{\prime\prime}$ if $M^\prime(x)\le M^{\prime\prime}(x)$, for all $x\in{S}$.
\end{definition}

For a multiset $M\in\mathcal{M}(S)$ and an element $x\in{S}$, we say that $x$ \textit{occurs} in $M$, and write $x\in{M}$, whenever $M(x)>0$.
Also, we say that the set $\{x\in{S}\mid{M}(x)>0\}$ is the \emph{support} of $M$, and write $supp(M)$.
A set $X\subseteq S$ can be viewed as a multiset by taking
\[
X(x)=
\begin{cases}
1, & x\in X,\\
0, & x\notin X.
\end{cases}
\]

In the rest of the paper, we rely on the standard operations on sets and multisets of union, intersection, and difference, respectively denoted for the two cases by $(\cup,\cap,\setminus)$ and $(\sqcup,\sqcap,\multisetminus)$.
We use the standard notation $\emptyset$ for the multiset with $\emptyset(x)=0$ for $x\in{S}$.

\section{Interval structures}\label{sec:S_intervals}
The models considered in this paper require a representation capable of expressing exact multiplicities, lower and upper bounds, threshold conditions, and families of admissible multisets. To this end, we introduce interval structures, which associate with each entity a range of admissible multiplicities. As will become apparent later, the same formalism can be used both to describe enabling conditions and to specify possible outcomes of transformations.

\begin{definition}[Interval structure]\label{def:int_struct}
An \emph{interval structure} over $S$ is a function
\[
W:S\rightarrow \mathbb{N}\times(\mathbb{N}^{+}\cup\{\top\}) 
\]
such that, for every $x\in S$, if $W(x)=(\ell_W(x),u_W(x))$, then $\ell_W(x)<u_W(x)$.
%
\end{definition}

For each \textit{entity} $x\in{S}$, an interval structure specifies its admissible multiplicity range. 

We now formalize the relationship between interval structures and multisets by defining when a multiset satisfies an interval structure. This will allow us to interpret every interval structure as a family of multisets.

\begin{definition}[Satisfaction]
Let $M\in\mathcal{M}(S)$ and let $W$ be an interval structure over $S$.
We say that $M$ \emph{satisfies} $W$, written $M\models W$, if $\ell_W(x)\le M(x)<u_W(x)$ for every $x\in S$.
\end{definition}

For an interval structure $W$, we denote by $\mathbf{M}(W)$ the family of multisets satisfying $W$:
\[
\mathbf{M}(W)
=
\{M\in\mathcal{M}(S)\mid M\models W\}.
\]

Interval structures provide a uniform representation for a variety of multiplicity-based objects. The following examples show that both sets and multisets arise as special cases. This observation will later allow both classical reaction systems and multiset reaction systems to be described within a common formalism.

\begin{example}[Sets]
Let $X\subseteq S$. Define the interval structure $W_X$ by
\[
W_X(x)=
\begin{cases}
(1,2), & \text{if }x\in X,\\
(0,1), & \text{if }x\notin X.
\end{cases}
\]

Then $\mathbf{M}(W_X)$ contains exactly one multiset, corresponding to $X$.
\end{example}

\begin{example}[Multisets]
Let $M\in\mathcal{M}(S)$. Define the interval structure $W_M$ by
\[
W_M(x)=(M(x),M(x)+1),
\qquad x\in S.
\]

Then $\mathbf{M}(W_M)=\{M\}$.
\end{example}

\section{Interval transformations}\label{S_transformations}
Interval structures provide a uniform way of describing families of multisets. We now use them to specify both the enabling conditions and the possible outcomes of transformations.

\begin{definition}[Interval-based transformation]\label{def:intBasedT}
An \emph{interval-based transformation} (short, \emph{transformation}) over $S$ is a pair $\tau=(\varepsilon_\tau,\pi_\tau)$, where $\varepsilon_\tau$ and $\pi_\tau$ are interval structures over $S$.
Then, 
%
%
$\varepsilon_\tau$ is called the \emph{enabling structure} of $\tau$, while $\pi_\tau$ is called its \emph{production structure}.
\end{definition}

The enabling structure specifies the family of states in which the transformation may be executed. The production structure specifies the family of multisets that may be produced when the transformation is applied.

Individual transformations describe isolated state changes. Computational models, however, are typically specified through collections of transformations that interact and may execute concurrently. We therefore introduce interval transformation systems as finite sets of interval transformations defined over a common background set.

\begin{definition}[Interval-based transformation system]\label{def:intBasedTS}
An \emph{interval-based transformation system} over $S$ is a pair $\mathcal{A}=(S,A)$, where $A$ is a finite set of interval transformations over $S$. 
The \emph{states} of an interval-based transformation system over $S$ are multisets over $S$.
\end{definition}

The enabling structure of a transformation specifies the family of states in which the transformation may occur. We now formalize this notion by defining when a state satisfies the requirements imposed by an enabling structure: a transformation is enabled whenever the current state belongs to the family of multisets represented by its enabling structure.

\begin{definition}[Enabling]\label{def:enabling}
Let $M\in\mathcal{M}(S)$ and let $\tau=(\varepsilon_\tau,\pi_\tau)$ be an interval-based transformation.
We say that $\tau$ is \emph{enabled} in $M$ if $M\models\varepsilon_\tau$.
Given a set of interval-based transformations $A$, the set of enabled transformations from $A$ in $M$ is denoted by $\EN(A,M)$.
\end{definition}

Reaction systems and their variants typically allow several reactions to occur simultaneously. 
To capture this form of concurrency, we group transformations into \textit{blocks}, which are interpreted as executing together. 
The multiplicity $B(\tau)$ in a block $B$ specifies how many times the transformation $\tau$ participates in the execution of the block.
However, not every collection of transformations can meaningfully execute simultaneously. Different transformations may impose incompatible enabling conditions on the same entities, making it impossible for their requirements to be satisfied in a common state. 
We therefore evaluate each block with respect to a notion of \textit{viability}, expressing that the enabling structures of its transformations are mutually compatible.

\begin{definition}[Block]\label{def:block}
Let $\mathcal{A}=(S,A)$ be an interval transformation system.
A \emph{block} in $\mathcal{A}$ is any multiset in $\mathcal{M}(A)$. 
The \emph{support} of $B$ is $\supp(B)=\{\tau\in A\mid B(\tau)>0\}$.

A block $B$ is \emph{viable} if $\bigcap_{\tau\in \supp(B)}\mathbf M(\varepsilon_\tau)\neq\emptyset$. 
Equivalently,
\[
\max_{\tau\in \supp(B)}
\ell_{\varepsilon_\tau}(x)
<
\min_{\tau\in \supp(B)}
u_{\varepsilon_\tau}(x), \forall x\in S.
\]
We denote the set of viable blocks in $\mathcal{A}$ by $\mathcal{B}(A)$.
\end{definition}

\begin{example}\label{ex:inconsistent}
Consider two transformations with enabling structures $\varepsilon_1(x)=(2,3),\varepsilon_2(x)=(3,5)$, respectively.
Obviously, no multiset satisfies both requirements simultaneously. Hence
$\mathbf{M}(\varepsilon_1)\cap\mathbf{M}(\varepsilon_2)=\emptyset$, and so, any block whose support contains both transformations is not viable.
On the other hand, with $\varepsilon_1(x)=(2,4),\varepsilon_2(x)=(3,5)$, we have $\mathbf{M}(\varepsilon_1)\cap\mathbf{M}(\varepsilon_2)=\{3\}$, so that a block constituted of the two transformations would be viable.
\end{example}

\begin{remark}\label{R1}
For a block $B\in\mathcal{B}(A)$, the enabling structures of its transformations induce an aggregate enabling structure $\varepsilon_B$, defined by
\[
\ell_{\varepsilon_B}(x)
=
\max_{\tau\in \supp(B)}
\ell_{\varepsilon_\tau}(x),
\qquad 
u_{\varepsilon_B}(x)
=
\min_{\tau\in \supp(B)}
u_{\varepsilon_\tau}(x),
\qquad
\forall x\in S.
\]
A state $M$ satisfies $\varepsilon_B$ if and only if it satisfies the enabling structure of every transformation in $\supp(B)$. 
Equivalently,
\[
\mathbf M(\varepsilon_B)
=
\bigcap_{\tau\in \supp(B)}
\mathbf M(\varepsilon_\tau).
\]
Thus, a block is viable precisely when $\mathbf M(\varepsilon_B)\neq\varnothing$.
The aggregate enabling structure $\varepsilon_B$ plays a role analogous to the combined reaction introduced in \cite{genova_enabling_2025}. In particular, $\varepsilon_B$ provides a canonical representation of the collective enabling requirements of a block and may serve as a basis for studying notions of enabling equivalence between blocks.
\end{remark}

Unlike classical reaction systems, where a set of reactions admits only finitely many subsets, blocks in the present framework are multisets, and transformations in them may therefore occur with arbitrary multiplicities. 
As a consequence, the collection of viable blocks need not be finite and is indeed infinite for any nonempty system, as any block with $supp(B)=\{\tau\}$ for some transformation $\tau$ can actually be formed by an arbitrary number of copies of $\tau$. 

The transformations in a block are interpreted as occurring simultaneously. Since production structures may represent multiple possible multisets, the application of a block may be non-deterministic. The semantics of block application will be defined in the next section.

\section{Block semantics}\label{S_policies}
Since a production structure does not necessarily specify a single multiset, but represents a family of possible outcomes, the application of a block may be non-deterministic.

The semantics of a block application consists of three stages. First, the transformations within the block contribute resource requirements, which are combined according to a chosen resource policy. Second, each transformation selects an outcome from its production structure and the resulting products are aggregated according to a chosen production policy. Finally, the aggregated products are used to update the current state according to a chosen update policy.

\subsection{Resource policies}\label{sub:resource_policies}
The enabling structure of a transformation determines the resources required for its execution. 
In particular, we define demand multisets, based on the lower bound of an enabling structure, to capture quantitative resource requirements. 
Upper bounds in enabling structures represent constraints rather than consumable resources and therefore do not contribute to demand.

Let $\tau=(\varepsilon_\tau,\pi_\tau)$ be an interval transformation.
The \emph{demand multiset} of $\tau$ is the multiset $D_\tau$ defined by $D_\tau(x)=\ell_{\varepsilon_\tau}(x)$ for every $x\in S$.

Different notions of resource sensitivity are obtained by combining the demand multisets of the transformations in a block in different ways.
%
A \emph{resource policy} associates with every block $B$ a multiset $\Delta(B)$, called the \emph{aggregated demand} of $B$.

Two natural resource policies are the following:
\begin{itemize}
\item \emph{additive demand policy}:   
\[\Delta^{+}(B)(x) = \bigsqcup_{\tau\in \supp(B)} B(\tau)D_\tau(x), \forall x\in S;\]
\item \emph{maximum demand policy}: 
\[\Delta^{\max}(B)(x) = \max_{\tau\in \supp(B)} D_\tau(x), \forall x\in S.\]
\end{itemize}

By convention, for $B$ with empty support the maximum function above returns 0. 

\begin{proposition}\label{prop:demand-comparison}
Let $B$ be a block. 
Then, for $x\in S$, we have $\Delta^{\max}(B)(x)\leq\Delta^{+}(B)(x)$.
\end{proposition}

\begin{proof}
For $x\in S$, the maximum of the values $D_\tau(x)$, with $\tau\in\supp(B)$, is bounded above by their sum.
\end{proof}

Resource policies determine how transformations interact when they require the same resources. 
Under the additive demand policy, resource requirements accumulate across
transformations. 
Consequently, transformations compete for available resources, and the execution of one transformation may prevent the simultaneous execution of another. 
In contrast, the maximum demand policy allows resources to be shared among transformations.
Multiple transformations may simultaneously rely on the same resources, and the demand of a block is determined by the largest individual requirement for each entity.

Viability only guarantees that the enabling structures of the transformations in a block are mutually compatible, but a viable block may still not be executable under a given resource policy. 
For example, additive demand may require more resources than are allowed by the upper bounds 
%
%
in the aggregate enabling structure. 
This motivates an intermediate notion, capturing whether a viable block can satisfy the requirements imposed by a resource policy.

\begin{definition}[$\Delta$-feasible block]\label{def:feasible}
Let $\Delta$ be a resource policy and let $B$ be a viable block. We say that $B$ is $\Delta$-feasible if there exists a state M such that

\[
M\models\varepsilon_B
\quad \text{and}\quad
\Delta(B)(x)\le M(x),
\qquad
\forall x\in S.
\]
\end{definition}


The notion of feasibility used above is induced by a resource policy: a viable block is feasible when its aggregate resource demand can be satisfied by some state compatible with its aggregate enabling structure.
Other variants may impose additional structural admissibility constraints on blocks. 
Such constraints are independent of resource availability and are best treated as separate feasibility requirements. 
One example is the mutual-exclusion constraint used in restricted reaction systems, discussed in Section \ref{S_restricted}.

Resource policies induce feasibility through aggregated demand, as in Definition \ref{def:feasible}.
Other feasibility policies may instead impose structural constraints on blocks.


Multiple occurrences of the same transformation may be limited by the upper bounds imposed by the aggregate enabling structure. 
Under additive demand, each occurrence contributes its demand, while the enabling constraints are determined by the support of the block. 
Hence, repeated occurrences of a transformation can yield infinitely many $\Delta^+$-feasible blocks only when the resources required by that transformation are not bounded above.

\begin{proposition}\label{prop:infiniteConsistent}
    Given an interval transformation system $\mathcal{A}=(S,A)$, the set of its viable blocks that are $\Delta^+$-feasible is infinite if and only if there exists $\tau\in{A}$ such that, for all $x\in{S}$, if $\ell_{\varepsilon_\tau}(x) > 0$, then $u_{\varepsilon_\tau}(x) =\top$.
\end{proposition}

\begin{proof}
Assume first that such a transformation $\tau$ exists. For every $n\geq 1$, define the block $B_n$ by $B_n(\tau)=n$ and $B_n(\sigma)=0$  for every $\sigma\neq\tau$. 
Under additive demand, we have $\Delta^+(B_n)(x)=n\ell_{\varepsilon_\tau}(x)$.
Then:
\begin{itemize}
    \item If $\ell_{\varepsilon_\tau}(x)=0$,we choose $M(x)=0$.
    \item If $\ell_{\varepsilon_\tau}(x)>0$, then we have $u_{\varepsilon_\tau}(x)=\top$ and we can choose $M(x)=n\ell_{\varepsilon_\tau}(x)$.
\end{itemize}

With these choices, for every $x\in S$, we have $\ell_{\varepsilon_\tau}(x)\le{M}(x)<u_{\varepsilon_\tau}(x)$ and $\Delta^+(B_n)(x)\le M(x)$.
Therefore $B_n$ is $\Delta^+$-feasible for all $n\geq 1$.

Conversely, suppose that there are infinitely many $\Delta^+$-feasible blocks. 
Since $A$ is finite, there exists a transformation $\tau\in A$ whose multiplicity is unbounded among these blocks. 
Let $x\in S$ be such that $\ell_{\varepsilon_\tau}(x)>0$. 
If $u_{\varepsilon_\tau}(x)\neq\top$, then $\Delta^+$-feasibility of $B$ requires the existence of a state $M$ such that
\[
B(\tau)\ell_{\varepsilon_\tau}(x)\le M(x)<u_{\varepsilon_\tau}(x).
\]
Since $u_{\varepsilon_\tau}(x)$ is finite, this is possible for only finitely many values of $B(\tau)$, contradicting unboundedness. 
Therefore $u_{\varepsilon_\tau}(x)=\top$.
\end{proof}

The maximum demand policy imposes no additional structural restrictions beyond viability, as shown in the following result.

\begin{proposition}
A block is viable if and only if it is
$\Delta^{\max}$-feasible.
\end{proposition}

\begin{proof}
If $B$ is viable, then there exists $M\models\varepsilon_B$.
For every $x\in S$,

\[
\Delta^{\max}(B)(x)
=
\max_{\tau\in\supp(B)}
\ell_{\varepsilon_\tau}(x)
=
\ell_{\varepsilon_B}(x)
\le M(x).
\]
Hence $B$ is $\Delta^{\max}$-feasible.

Conversely, every $\Delta^{\max}$-feasible block admits a state satisfying
$\varepsilon_B$, and therefore is viable.
\end{proof}

A $\Delta$-feasible block might still not be executable in a given state. 
To execute simultaneously, the transformations of the block must not only be mutually compatible, but must also be individually enabled and have access to sufficient resources under the adopted resource policy in the current state. 
This leads to the following notion of enabling.

\begin{definition}[$\Delta$-enabled block]
Let $M$ be a state and let $\Delta$ be a resource policy.
A block $B$ is \emph{$\Delta$-enabled} in $M$, under the $\Delta$ policy, if:

\begin{itemize}
\item $B$ is $\Delta$-feasible; 
\item $M\models\varepsilon_B$;
\item $\Delta(B)(x)\le{M}(x)$, $\forall x\in S$.
\end{itemize}
A 
%
%
block $B$ is $\Delta$-\emph{maximal} in $M$ if
it is $\Delta$-enabled in $M$ and 
%
%
for no $\Delta$-enabled block $B^\prime$ we have
$B\sqsubset{B}^\prime$.
\end{definition}

\begin{proposition}
Let $B$ be a block, $\Delta$ a resource policy, and $M$ a state. If $B$ is $\Delta$-enabled in $M$, then $B\in \mathcal M(\mathrm{EN}(A,M))$.
\end{proposition}

\begin{proof}
Since $B$ is $\Delta$-enabled in $M$, we have $M\models\varepsilon_B$.
By Remark \ref{R1}, this means that $M\models\varepsilon_\tau$ for every $\tau\in\supp(B)$. Hence every transformation in $\supp(B)$ is enabled in $M$, that is, $\supp(B)\subseteq \mathrm{EN}(A,M)$.
Therefore, $B$ is a multiset over $\mathrm{EN}(A,M)$, and so $B\in \mathcal M(\mathrm{EN}(A,M))$.
\end{proof}

Corollary \ref{cor:additive-implies-maximum-enabled} is a simple consequence of how additive and maximum demand were defined. 

\begin{corollary}\label{cor:additive-implies-maximum-enabled}
If block $B$ is $\Delta^+$-enabled in state $M$, then it is $\Delta^{\max}$-enabled in $M$.
\end{corollary}

\begin{proof}
Assume that $B$ is $\Delta^+$-enabled in $M$. Then, by definition, $M\models \varepsilon_B$ and
$\Delta^+(B)(x)\le M(x)$ for every $x\in S$. By Proposition \ref{prop:demand-comparison}, $\Delta^{\max}(B)(x)\le \Delta^+(B)(x)$
for every $x\in S$. Hence $\Delta^{\max}(B)(x)\le M(x)$
for every $x\in S$. Therefore $B$ is $\Delta^{\max}$-enabled in $M$.
\end{proof}

\begin{example}\label{ex:add_vs_max}
    Let $\tau_1,\tau_2$ be two interval transformations over $S=\{a\}$ such that $\varepsilon_1(a)=(2,5),\varepsilon_2(a)=(3,5)$.
    Then, with $M\in\mathcal{M}(S)$ given by $M(a)=4$, the block $B_1$,  given by $B_1(\tau_1)=B_1(\tau_2)=1$, is enabled in $M$ under maximum demand, but not under additive demand.
    Conversely, $B_2$ given by $B_2(\tau_1)=2,B_2(\tau_2)=0$ is enabled in $M$ under both policies.
\end{example}

Note that even a block $B$ consisting of a single occurrence of an individual transition $\tau$ can be a $\Delta^+$-maximal block in $M$, if we have $2\ell_{\varepsilon_\tau}(x)>M(x)$ for some $x\in{S}$, and either 
$u_{\varepsilon_\tau}(x)<\ell_{\varepsilon_{B^\prime}}(x)$ or $\ell_{\varepsilon_\tau}(x)\ge{u}_{\varepsilon_{B^\prime}}(x)$ for some $x\in{S}$, for any other block $B^\prime$. 
On the other hand, no $\Delta^{max}$-enabled block in $M$ can be considered $\Delta^{max}$-maximal in $M$, for any $M\in\mathcal{M}(S)$.

\subsection{Production policies}\label{sub:production_policies}

Once a block has been selected, the production semantics proceeds in two stages.

First, each occurrence of a transformation in the block independently selects a multiset from the family represented by its production structure. 
Since blocks are multisets, multiple copies of the same transformation may participate in the execution of a block.
These copies are treated independently and may therefore select different production multisets. 

Second, the selected outcomes are combined according to a chosen production
policy to obtain the aggregated production of the block.

%
An \emph{outcome selection} for block $B$ assigns a multiset $P_{\tau}(i)\in\mathbf{M}(\pi_\tau)$
to every $\tau\in\operatorname{supp}(B)$ and every $1\le{i}\le{B}(\tau)$.
%
The selected outcomes are combined according to a production policy.
%
A \emph{production policy} associates with every outcome selection $(P_\tau(i))_{\tau\in \supp(B),1\le i\le B(\tau)}$ a multiset $\Pi(B)$, called the \emph{aggregated production} of $B$.

Two natural production policies are the following:
\begin{itemize}
\item \emph{additive production:}
\[
\Pi^{+}(B)(x)
=
\bigsqcup_{\tau\in \supp(B)}\bigsqcup_{i=1}^{B(\tau)} P_{\tau}(i)(x), \forall x\in S;
\]
\item \emph{maximum production:}
\[
\Pi^{\max}(B)(x)
=
\max\{P_{\tau}(i)(x) \mid 
\tau\in \supp(B), 1\le i\le B(\tau)\}, \forall x\in S.
\]
\end{itemize}
By convention, the maximum of an empty block is set to 0. 

When all multisets are Boolean, maximum production coincides with set union.

\subsection{Update policies}\label{sub:update_policies}
The aggregated production obtained from a block is used to construct the next state.
%
An \emph{update policy} specifies how a state $M$ is transformed into a successor state $M^\prime$ using the aggregated demand $\Delta(B)$ and aggregated production $\Pi(B)$ of a globally enabled block $B$.

Two natural update policies are the following:

\begin{itemize}
\item \emph{destructive update policy}: replaces the current state by the aggregated production  $M^\prime=\Pi(B)$;
\item \emph{conservative update policy}: removes the resources consumed by the block and adds the aggregated production $M^\prime=(M-\Delta(B))+\Pi(B)$.
\end{itemize}

The semantics of the framework is determined by three independent design choices: a resource policy, a production policy, and an update policy. Different combinations of these choices give rise to different variants of reaction systems and resource-sensitive computational models.

\section{Execution strategies}\label{S_execution}
The semantics introduced in the previous section determines the effect of executing a block. A computation model must additionally specify which blocks may be executed in a given state. Different choices lead to different execution strategies and different forms of concurrency.

Let $\mathcal A=(S,A)$ be an interval transformation system, $\Delta$ a resource policy, $M$ a state, and $B$ be a $\Delta$-enabled block in state $M$.
An execution step of block $B$ in state $M$ consists of:
\begin{itemize}
\item selecting an outcome $P_\tau\in\mathcal M(\pi_\tau)$ for each transformation $\tau\in \supp(B)$, selecting $B(\tau)$ outcomes from $\mathcal{M}(\pi_\tau)$;
\item computing the aggregated production according to the selected production policy;
\item constructing the successor state according to the selected update policy.
\end{itemize}

Different execution strategies may impose additional restrictions on the choice of which blocks to apply out of the enabled ones:
\begin{itemize}
\item \emph{sequential strategy}: restrict execution to singleton blocks;
\item \emph{concurrent strategy}: execution of 
any one of the enabled blocks;
%
%
\item \emph{maximally concurrent strategy}: execution of maximal enabled blocks.
\end{itemize}

Sequential, concurrent, and maximally concurrent strategies provide different interpretations of the same underlying transformation system. 
It is to be noted that any block can be seen as the multiset union of singleton blocks, so that the sequential strategy can be seen as a particular case of the concurrent strategy, as, on the other end, can the maximally concurrent one.
The choice of strategy affects the structure of the reachable state space and the degree of non-determinism exhibited by computations.

Classical reaction systems correspond to a maximally concurrent strategy in which the unique maximal enabled block is $\EN(\mathcal A,M)$ itself, so that every enabled transformation participates in the computation step.


The framework therefore distinguishes three independent sources of variation:

\begin{itemize}
\item \emph{semantic choices}, represented by resource, production, and update policies;
\item \emph{execution choices}, represented by different strategies for selecting enabled blocks;
\item \emph{outcome choices}, arising from the non-deterministic production structures associated with transformations.
\end{itemize}

This separation makes it possible to study concurrency, resource sensitivity, and non-deterministic production independently within a common formal framework.

\begin{definition}[Execution step]\label{def_execution}
Let $\mathcal{F}=(\mathcal{A},\Delta,\Pi, U,E)$ be a framework, where $\mathcal A$ is an interval transformation system, $\Delta$ is a resource policy, $\Pi$ is a production policy, $\mathbf{U}$ is an update policy, and $\mathbf{E}$ is an execution strategy. 
Let $M$ be a state.

If $B$ is a $\Delta$-enabled block selected according to 
%
%
$\mathbf{E}$, 
and $M^\prime$ is the state obtained by applying the corresponding 
%
%
policies 
$\Pi$, $\mathbf{U}$
to $B$, then we write $M\xrightarrow{B}_{\mathcal{F}}M^\prime$ and say that $M^\prime$ is obtained from $M$ by an execution step of $B$ 
according to $\mathcal{F}$.

Furthermore, we write $M\Longrightarrow{M}^\prime_{\mathcal{F}}$ if there exists a block $B$ such that $M\xrightarrow{B}_{\mathcal{F}}M^\prime$.
\end{definition}

In the following, we will drop the subscript ${}_{\mathcal{F}}$ and simply write $M\xrightarrow{B}M^\prime$ when ${\mathcal{F}}$ is understood in the context.

\section{A unified view of reaction system variants}\label{S_variants}
The framework introduced in the previous sections provides a common semantic foundation for a variety of reaction-system variants. 
We now show how several representative models can be recovered within the framework, thereby providing a unified view of their underlying semantics. 
Throughout this section, we say that a reaction-system variant is \emph{recovered} by the framework if, up to the encoding of states and reactions described below, every computation step of the original model coincides with the computation step induced by the corresponding choice of resource, production, update, and execution policies.
In other words, the operational semantics of the original model coincides with that induced
by the chosen policies.

\begin{definition}[Recovery]\label{def:recovery}
Let $\mathcal R$ be a reaction-system variant with state space $\mathcal{S_R}$.
Let $\mathcal{F}=(\mathcal{A},\Delta,\Pi, U,E)$ be a framework, where $\mathcal A=(S,A)$ is an interval transformation system, $\Delta$ is a resource policy, $\Pi$ is a production policy, $\mathbf{U}$ is an update policy, and $\mathbf{E}$ is an execution strategy. 
Assume that states of $\mathcal{R}$ are encoded by a function
\[
\iota_S:\mathcal {S_R}\to\mathcal{M}(S),
\]
and that the defining reactions of $\mathcal R$ are encoded as transformations of $\mathcal{A}$.

We say that $\mathcal{R}$ is \emph{recovered} by the framework $\mathcal{F}$ if, for all states $M,M^\prime\in\mathcal{S_R}$,
\[
M\to_{\mathcal{R}}M^\prime
\quad\Longleftrightarrow\quad
\iota_S(M)\Longrightarrow_{\mathcal F}\iota_S(M^\prime).
\]
Here $\to_{\mathcal R}$ denotes the one-step evolution relation of the original variant, and $\Longrightarrow_{\mathcal{F}}$ is the transition relation induced by Definition \ref{def_execution} for framework $\mathcal{F}$.
\end{definition}

\subsection{Classical reaction systems}\label{sub:classical}
Classical reaction systems from \cite{ehrenfeucht_reaction_2007} are expressed in our framework by restricting admissible blocks to Boolean multisets and using the maximally concurrent strategy.

\begin{construction}\label{cons:classicalEncoding}
  Let $\tau=(R_\tau,I_\tau,P_\tau)$ be a reaction over a background set $S$.
We associate with $\tau$ the interval transformation
\[
\tau^\prime=(\varepsilon_\tau,\pi_\tau),
\]
where
\[
\varepsilon_{\tau^\prime}(x)=
\begin{cases}
(1,\top), & x\in R_\tau,\\
(0,1), & x\in I_\tau,\\
(0,\top), & \text{otherwise},
\end{cases}
\]
and
\[
\pi_{\tau^\prime}(x)=
\begin{cases}
(1,2), & x\in P_\tau,\\
(0,1), & x\notin P_\tau.
\end{cases}
\]  
\end{construction}

The condition $M\models\varepsilon_\tau$
is equivalent to the standard enabling condition requiring all reactants to be present and all inhibitors to be absent.
Since products are Boolean multisets, maximum production coincides with set union. 
Furthermore, maximum demand reflects the fact that resources are not consumed and may simultaneously support multiple reactions.

\begin{proposition}\label{prop:classicalRS}
Classical reaction systems are recovered by a framework $\mathcal{F}$ combining $\Delta^{\max},\Pi^{\max}$, 
%
%
and destructive update under the maximally concurrent execution strategy.
\end{proposition}

\begin{proof}
Let $\mathcal{R}=(S,A)$ be a classical reaction system and let $W\subseteq S$ be a state. 
For each reaction $\tau=(R_\tau,I_\tau,P_\tau)\in{A}$, let $\tau^\prime=(\varepsilon_{\tau^\prime},\pi_{\tau^\prime})$ be its encoding according to Construction \ref{cons:classicalEncoding}.

For every state $W$, we have:
\[ 
W\models\varepsilon_{\tau^\prime} \quad\Longleftrightarrow\quad{R}_\tau\subseteq{W}\text{ and }I_\tau\cap W=\emptyset.
\]
Hence, the enabled transformations in the encoded system are exactly the
encodings of the classically enabled reactions.

Under the maximally concurrent strategy, the executed block $B$ has as support precisely the set of enabled transformations $\EN(A,W)$. 
Since products are Boolean, maximum production coincides with set union, and therefore
\[
\Pi^{\max}(B)=\iota_S\left(\bigcup_{\tau\in\EN(A,W)} P_\tau\right).
\]
By destructive update, the successor state is
\[
\iota_S(W')=\Pi^{\max}(B) = \iota_S\left(\bigcup_{\tau\in\EN(A,W)}P_\tau\right),
\]
which is exactly the classical reaction-system result. 
Thus the execution steps according to $\mathcal{F}$ coincide under the given encoding; hence, the classical model is recovered.
\end{proof}

\subsection{Restricted reaction systems}\label{S_restricted}
Restricted reaction systems \cite{aman_mutual_2020, aman_reaction_2025} introduce nondeterminism by relaxing
the maximally concurrent execution strategy of classical reaction systems.
Instead of requiring all enabled reactions to participate in a computation step, the executed set of reactions may be constrained by a mutual-exclusion relation $\#$ between reactions.

Within our framework, restricted reaction systems are represented by adding the feasibility policy $\Delta^\#$. 



\begin{definition}[$\Delta^\#$-feasible block]
Let $\mathcal A=(S,A)$ be an interval transformation system, and let
$\#\subseteq A\times A$ be a symmetric incompatibility relation.
A viable block $B\in\mathcal B(A)$ is $\Delta^\#$-feasible if
$(\tau_1,\tau_2)\notin \#$
for all distinct transformations
$\tau_1,\tau_2\in\operatorname{supp}(B)$.
%
\end{definition}

Thus, $\Delta^\#$-feasibility is a structural feasibility condition: it depends only on the support of the block and not on the current state or on resource availability.
Restricted reaction systems are obtained by requiring executable blocks to be both $\Delta^{\max}$-feasible and $\Delta^\#$-feasible, while keeping maximum production and destructive update under a concurrent execution strategy.

\subsection{Multiset reaction systems}\label{sub:multiRS}
Multiset reaction systems as defined in \cite{bottoni_multiset_2025} generalize reaction systems by allowing arbitrary multiplicities.
We consider here the concurrent operational semantics adopted in \cite{BLP26}. 
We briefly recall the concurrent destructive semantics of multiset reaction
systems needed for the present comparison.  

A multiset reaction is a triple $\tau=(R_\tau,I_\tau,P_\tau)$, where reactants, inhibitors, and products are multisets and $R_\tau(x)<I_\tau(x)$ for all $x\in \supp(I_\tau)$. 
A reaction is enabled in a multiset state $M$ whenever
\[
R_\tau(x)\le M(x),\ \  \forall x\in \supp(R_\tau), \qquad\text{and}\qquad
M(x)<I_\tau(x),\ \  \forall x\in \supp(I_\tau).
\]
A block $B$ is a multiset of enabled reactions such that its total reactant multiset $R_B$ satisfies $R_B\sqsubseteq M$. Executing $B$ produces the total product multiset $P_B$, which becomes the successor state under the destructive update semantics considered here. We refer to~\cite{bottoni_multiset_2025, BLP26} for a more detailed presentation of the model and its properties.

For a multiset reaction $\tau=(R_\tau,I_\tau,P_\tau)$, the multiplicities appearing in the reactant and product multisets are represented directly through the lower bounds of the corresponding interval structures. 
The enabling structure specifies the required multiplicities, while the production structure specifies the multiplicities contributed by the reaction.
The semantics of multiset reaction systems is characterized by resource competition and product accumulation. Consequently, additive demand and additive production are used. Since the current state is replaced by the generated products after each step, destructive update is adopted.


\begin{proposition}\label{prop:semMRS}
The operational semantics of multiset reaction systems described above is recovered by 
%
%
a framework $\mathcal{F}$ combining $\Delta^+,\Pi^+$, and destructive update
under the concurrent execution strategy.
\end{proposition}

\begin{proof}
Let $M$ be a multiset state and let
$\tau=(R_\tau,I_\tau,P_\tau)$ be a multiset reaction.

By construction of the encoding,
\[
M\models\varepsilon_\tau
\iff
R_\tau(x)\le M(x)\ \forall x\in\supp(R_\tau),
\quad\text{and}\quad
M(x)<I_\tau(x)\ \forall x\in\supp(I_\tau).
\]
so the enabling condition of individual reactions is preserved.

Let $B$ be a concurrently selected block. By definition of additive demand,
\[
\Delta^+(B)
=
\bigsqcup_{\tau\in\operatorname{supp}(B)}
B(\tau)\otimes R_\tau
=
R_B.
\]
Hence,
\[
\Delta^+(B)\sqsubseteq M
\iff
R_B\sqsubseteq M,
\]
so the enabling condition of blocks is preserved.

Similarly, additive production gives
\[
\Pi^+(B)
=
\bigsqcup_{\tau\in\operatorname{supp}(B)}
B(\tau)\otimes P_\tau
=
P_B.
\]

Finally, destructive update sets the successor state to
\[
M'=\Pi^+(B)=P_B,
\]
which is exactly the successor state in the operational semantics described
above. Since both models employ the concurrent execution strategy, the
execution steps coincide under the encoding. Therefore, the operational
semantics of multiset reaction systems is recovered.
\end{proof}

Under this semantics, simultaneously executed transformations compete for resources and contribute cumulatively to the successor state.

Several execution strategies have been considered in \cite{bottoni_multiset_2025}, such as sequential, concurrent or constrained concurrent, respectively indicating that only one reaction is enabled at a time, any arbitrary feasible block is chosen for execution, or that limits on the number of instances of reactions are used to select the block to be executed. 

\subsection{Reaction systems with concentration}\label{sub:concentration}
A different take on the extension of reaction systems to the usage of multisets is in \cite{MKP16}, where \textit{reaction systems with discrete concentration} are defined on multisets (\textit{bags} in the paper), taking $\mathbf{t}(x)$ as an indication of the \textit{concentration}  of the entity $x\in{S}$ in a state $\mathbf{t}$.
Rather than the union of multisets, they consider the \textit{merge} of a set of bags $\mathbf{G}$ (noted $\merge\mathbf{G}$) by taking, for each entity $x\in{S}$, the maximum of the multiplicities with which $x$ appears in each bag in $\mathbf{G}$ (and $0$ if it does not appear in any bag).

Then, a $c$-reaction $a=(\mathbf{r}_a,\mathbf{i}_a,\mathbf{p}_a)$ is composed of three bags, with $\mathbf{r}(x)<\mathbf{i}(x)$ for $x\in\operatorname{supp}(\mathbf{i}_a)$, and $a$ is enabled in $\mathbf{t}$ if $\mathbf{r}_a(x)\le\mathbf{t}(x)$ for $x\in\operatorname{supp}(\mathbf{r}_a)$ and $\mathbf{t}(x)<\mathbf{i}_a(x)$ for $x\in{supp}(\mathbf{i}_a)$.
The new state after applying the $c$-reactions in $A$ is then given by $\merge\{\mathbf{p}_a\mid{a}\in\EN(A,\mathbf{t})\}$.

\begin{proposition}\label{prop:concentration}
    The operational semantics of reaction systems with discrete concentrations is recovered by 
    a framework $\mathcal{F}$ combining $\Delta^{\max},\Pi^{\max}$,
    %
    %
    and destructive update under the maximal concurrent execution strategy.
\end{proposition}

\begin{proof}
    The proof is analogous to that for Proposition \ref{prop:semMRS}, by replacing the arguments for additive demand and production with those for maximum demand and production, as in the proof of Proposition \ref{prop:classicalRS}.
    Maximum demand reproduces the enabling semantics of reaction systems with concentrations, while maximum production coincides with the merge operator. Destructive update replaces the current state by the merged products, exactly as in the original definition.
\end{proof}

\subsection{Resource-preserving multiset reaction systems}\label{sub:preserving}
A natural extension of multiset reaction systems is obtained by combining explicit resource competition with state persistence, as proposed in \cite{BLP26}.
As in multiset reaction systems, resources and products are aggregated additively. However, the update policy preserves resources that are not consumed during the execution of a block.

\begin{proposition}
The resource-preserving multiset reaction systems of \cite{BLP26} are recovered by a framework $\mathcal{F}$  combining $\Delta^+$, $\Pi^+$, and conservative update under the concurrent execution strategy.
\end{proposition}

\begin{proof}
The encoding of reactions and the enabling condition are the same as for multiset reaction systems from Section \ref{sub:multiRS}. 
Thus additive demand computes
\[
\Delta^+(B)(x)
=
\sum_{\tau\in\operatorname{supp}(B)} B(\tau)R_\tau(x),
\]
and additive production computes
\[
\Pi^+(B)(x)
=
\sum_{\tau\in\operatorname{supp}(B)} B(\tau)P_\tau(x).
\]
The only difference is the update policy. Under conservative update, the successor state is
\[
M'=(M-\Delta^+(B))+\Pi^+(B).
\]
Thus, the reactants consumed by the block are removed, the products are added, and all non-consumed resources persist. This is exactly the resource-preserving transition of $[14]$. Since the execution strategy is concurrent in both models, the execution steps coincide under the encoding.
\end{proof}

This semantics simultaneously captures quantitative resource consumption, accumulation of products, and persistence of resources across transitions.

The examples above demonstrate that several established reaction-system variants can be represented within a common semantic framework, as summarised in Table \ref{tab:instantiations}. 
By separating resource management, product aggregation, and state updates into independent components, the framework exposes the assumptions underlying existing models and provides a systematic basis for the development of new reaction-system-inspired computational formalisms.

\begin{table}[htb]
\caption{Reaction-system variants recovered within our framework.}
\label{tab:instantiations}
\centering
\begin{tabular}{lllll}
\hline
Model & Resource policy & Production policy & Update policy & Execution\\
\hline
Classical RS & $\Delta^{\max}$ & Maximum & Destructive & Maximal\\
Restricted RS & $\Delta^{\max}$, $\Delta^\#$ & Maximum & Destructive & Concurrent\\
Multiset RS & $\Delta^+$ & Additive & Destructive & Concurrent\\
RS with concentration & $\Delta^{max}$ & Maximum & Destructive & Maximal \\
Resource-preserving mRS & $\Delta^+$ & Additive & Conservative & Concurrent\\
\hline
\end{tabular}
\end{table}

The aggregate enabling structure associated with a block also establishes a connection with the notions of combined reactions, enabling equivalence, and cover relations studied in \cite{genova_enabling_2025}. Since $\varepsilon_B$ characterizes precisely the states that enable all transformations in a block, it provides a natural representation of the collective enabling behaviour of the block. A detailed investigation of the relationship between block-based semantics and the equivalence notions introduced in \cite{genova_enabling_2025} is left for future work.

\subsection{Quantitative reaction systems}\label{sub:quantitative}
Quantitative reaction systems (qRS) were introduced in \cite{mitrana_quantitative_2025}. Unlike the variants discussed above, quantitative reaction systems require a preprocessing stage that transforms the current state before enabling is evaluated. Consequently, they cannot be obtained solely through a choice of resource, production, and update policies. Instead, they require an additional semantic component that can be incorporated naturally into the framework.  

A quantitative reaction is a tuple $\alpha=(R,I,X,P)$, where $R,I,P$ denote the reactants, inhibitors, and products, respectively, while $X\subseteq R\times I$ specifies which inhibitors act on which reactants. 
%
States are here represented by multisets rather than sets.
The semantics of qRS is based on a preprocessing phase applied before enabling is evaluated. Given a multiset of reactions
%
%
$\Delta=\alpha^{i_1}_1\dots\alpha^{i_k}_k\in\mathcal{M}(\mathbf{R}(S))$
and a state $M\in\mathcal{M}(S)$, the qRS semantics repeatedly removes matching reactant--inhibitor pairs according to the relations $X$ associated with the reactions in $\Delta$. This process yields a transformed state $M_\Delta$. The block is then enabled if the aggregate reactant requirements remain available in $M_\Delta$ while the aggregate inhibitors have been eliminated. The successor state is obtained by adding the products of the reactions in $\Delta$ and retaining the resources that remain after the reactants have been consumed.

\begin{example}
Consider the quantitative reaction (defined as in \cite{mitrana_quantitative_2025}) $\alpha=(R,I,X,P)$, where
\[
R=\{a,c\}, \qquad
I=\{b,d\}, \qquad
X=\{(a,b)\}.
\]

Let the current state be the multiset $M = \{3a, 2b, c\}$.
The preprocessing phase repeatedly removes matching copies of $a$ and $b$ since $X=\{(a,b)\}$. Since two copies of $b$ are present, two copies of $a$ are cancelled as well, while the only copy of $c$ is not consumed, yielding the transformed state $M_\alpha=\{a,c\}$.
The inhibitor $b$ has been completely eliminated (and the inhibitor $d$ was not present from the start), while one copy for each of the reactants $a,c$ remains available. Consequently, the enabling condition is satisfied. 

In contrast, if $M' = \{2a,3b,c\}$, then the preprocessing phase yields $M'_\alpha=b$, and the reaction is not enabled because no copy of $a$ remains available after the cancellation process.
\end{example}

This behaviour of qRS can be incorporated into the present framework by introducing a preprocessing policy $\rho(B,M)$,
which associates with every block $B$ and state $M$ a transformed state. For quantitative reaction systems, $\rho(B,M)$ corresponds precisely to the state $M_\Delta$ produced by the cancellation/inhibition procedure of \cite{mitrana_quantitative_2025}. 
The enabling conditions, resource requirements, and update policies are then evaluated on $\rho(B,M)$ rather than on the original state $M$.
Hence, quantitative reaction systems can be viewed as an extension of the proposed framework obtained by augmenting the semantic pipeline with a preprocessing stage:
\[
M
\;\longrightarrow\;
\rho(B,M)
\;\longrightarrow\;
\text{enabling}
\;\longrightarrow\;
\text{production}
\;\longrightarrow\;
\text{update}.
\]
This points to the possibility of integrating additional semantic mechanisms without modifying the core notions of interval structures, transformations, blocks, and execution strategies.

\subsection{Simulation techniques}\label{sub:2-step}
Literature on reaction systems has often considered simulations between variants, either describing how the semantics of variant can be reduced to  that of the classical version, or, vice versa, how particular restrictions on the form of reactions still maintain the expressive power of standard reaction systems. 

In most cases, this is achieved by defining the simulating system over some extension $S^\prime$ of the original background set $S$ for the simulated system, and then comparing the obtained states at steps which satisfy some condition, typically that of being defined over $S$ only. 

Usually, this results in so called 2-step simulations, where states are  alternately formed with entities from either of $S^\prime,S$, the states over $S$ being produced at even steps. 
Such a technique is used for simulation through minimal reactions (see, e.g., \cite{salomaa_minimal_2013,manzoni_reaction_2013,teh_simulation_2020}) and for realising transactional behaviours where the result of applying a set of reactions at odd steps (producing states over $S^\prime$) can result in commitment or rollback to states in $S$ (see \cite{bottoni_transactions_2021}.

In our approach, this would result in the addition of a \textit{filtering policy} after the completion of a process, to retain as effective steps only those producing states satisfying the constraint.
Note that in the considered cases, processes which are interrupted at an odd step would not provide a consistent simulation, and would therefore be discarded by the filtering policy.

\subsection{Beyond reaction systems: Petri nets and related models}\label{sub:multiPetri}
Petri nets (see e.g., \cite{Mur89}) provide a well-established formalism for modelling concurrent
systems based on multisets of tokens. Their operational semantics can also be
expressed within 
%
our
framework by interpreting places as entities,
markings as multisets, and transitions as 
%
%
transformations.

A Petri net consists of a finite set of places $P$, a finite set of transitions $T$, weighted pre- and post-incidence relations, and capacities $\gamma:P\rightarrow\mathbb N$.

A marking is a multiset $M\in\mathcal M(P)$
such that $M(p)\le\gamma(p)$, for all $p\in P$.
A transition is enabled if every input place contains the number of tokens 
required by the weight of the pre-incidence relation
and every output place has sufficient remaining capacity
to accommodate the number of tokens dictated by the weight of the post-incidence relation. 
Firing the transition removes the required input tokens and adds the produced output tokens.

Each marking is represented directly as a multiset over $P$. Each transition $\tau$ is encoded as an interval-based transformation $\tau'=(\varepsilon_{\tau'},\pi_{\tau'})$, where the lower bounds of $\varepsilon_{\tau'}$ represent the input arc weights, the production structure $\pi_{\tau'}$ represents the output arc weights, and the upper bounds encode the place capacities.
The capacity constraints are enforced through an additional feasibility policy $\Delta^\gamma$, while the conservative update reproduces the consumption and production of tokens.

\begin{proposition}
The operational semantics of Petri net firing is recovered by combining maximum demand, the capacity feasibility policy $\Delta^\gamma$, additive production, conservative update, and the sequential execution strategy.
\end{proposition}

\begin{proof}
The enabling structure of each encoded transformation requires every input place to contain at least the corresponding arc weight, while the capacity feasibility policy $\Delta^\gamma$ guarantees that firing the transition cannot exceed the capacity of any output place.
Maximum demand correctly evaluates enabling since each transition is considered independently under the sequential execution strategy. Additive production reproduces the addition of tokens to the output places, and conservative update removes the consumed input tokens while preserving all remaining tokens.
Therefore, every Petri net firing step coincides with one execution step of the framework under the encoding, and the other way around.
\end{proof}

The construction extends naturally to concurrent firing. A concurrently fired set of transitions is represented by a block of interval-based transformations. Input-token requirements are accumulated through additive demand, while the capacity feasibility policy $\Delta^\gamma$ is evaluated for the entire block. Products are aggregated additively and the marking is updated conservatively. Consequently, concurrent Petri net semantics is recovered by combining $\Delta^+$, $\Delta^\gamma$, additive production, conservative update, and the concurrent execution strategy.

Condition-event Petri nets correspond to the Boolean restriction in which all arc weights and capacities are equal to one. Additional feasibility constraints expressing contact and conflict are captured by a specialised feasibility policy $\Delta^{CE}$. The resulting semantics is recovered by combining $\Delta^{\max}$, $\Delta^{CE}$, maximum production, destructive update, and the concurrent execution strategy.

Finite-state automata (see e.g., \cite{hopcroft2006introduction}) arise as the special case of condition-event Petri nets in which every transition has exactly one input and one output place and each marking contains a single token. Their dynamics is therefore recovered by combining $\Delta^{\max}$, maximum production, destructive update, and the sequential execution strategy.

Similarly, unbounded multiset rewriting, corresponding to single-membrane P systems \cite{paun_bridging_2013} without membrane rewriting, is recovered by taking unbounded upper limits, additive demand, additive production, conservative update, and the maximally concurrent execution strategy.

\section{Conclusions}\label{sec:concls}
We have introduced an interval-based framework for representing enabling conditions and production outcomes in reaction-system-like models. 
Interval structures provide a common language for exact multiplicities, threshold constraints, saturation bounds, and non-deterministic production choices.

The main contribution of the framework is the decomposition of the operational semantics into independent resource, production, update, and execution policies, together with nondeterministic outcome selection.
This separation makes the assumptions underlying different variants of reaction systems explicit. 
In particular, it distinguishes resource competition, captured by additive demand, from resource sharing, captured by maximum demand.

We have shown how several existing models can be viewed as particular instantiations of our framework, including classical reaction systems, multiset reaction systems, and resource-preserving multiset reaction systems. 
%
%
More generally, the framework provides a systematic way of describing and comparing reaction-system variants through independent semantic choices. 
The framework can also be extended with additional semantic stages, as illustrated by the treatment of quantitative reaction systems through a preprocessing policy,
or of 2-step simulations through a post-mortem filtering of process steps.

By its separation of the different aspects of enabling and applying transformations, and by accommodating multisets in a natural way, the framework may support the definition of transformation laws mixing the different policies, beyond those considered in literature so far.

Moreover, the definition of the production policy is consistent with a ``multi-agent" view of transformations, in which each transformation is seen as an independent agent able to contribute a multiset out of the family of multisets defined by its production structure.

Several lines of research may therefore be identified, including the systematic study of reachability, equivalence, and verification problems within this unified setting, as well as the investigation of further resource and update policies motivated by applications in concurrency, natural computing, and biological modelling.

\bibliographystyle{unsrt}
\bibliography{RS}

\end{document}